\documentclass{article}

\usepackage[ruled]{algorithm2e}

\usepackage{amsmath,amsfonts,amssymb,bbm} 
\SetArgSty{textrm}  % for algorithm2e
\SetAlFnt{\small}
\SetAlCapFnt{\small}
\SetAlCapNameFnt{\small}
\SetAlCapHSkip{0pt}
\IncMargin{-\parindent}

\usepackage{fullpage}
\usepackage{amsthm}
\usepackage{pdfpages}
\usepackage{float}
\usepackage{eso-pic}
\usepackage{ifthen}
\usepackage{pstricks}
\usepackage{url}

\usepackage{mathrsfs}

\newtheorem{theorem}{Theorem}
\newtheorem{corollary}[theorem]{Corollary}
\newtheorem{lemma}[theorem]{Lemma}

{\theoremstyle{remark} }
{\theoremstyle{definition} \newtheorem{definition}[theorem]{Definition}}

\renewcommand{\qedsymbol}{\ensuremath{\blacksquare}}

\newcommand{\remove}[1]{}
\newcommand{\suppress}[1]{}

\newcommand{\RR}{{\mathbb{R}}}

\newcommand{\eps}{\epsilon}

\DeclareMathOperator{\best}{br}
\newcommand{\ljs}[1]{{ \textcolor{red}{(Leonard says:  #1)}}{}}

\begin{document}

\title{% Finite-Level Best-Response Dynamics in Fisher Markets \\
%Bounded Ply Best Response Dynamics in Fisher Markets
Market Dynamics of Best-Response with Lookahead
% in Fisher Markets
}
% with Inconsistent Beliefs}

\author{
Krishnamurthy Dvijotham\thanks{Caltech CMI, Engineering and 
Applied Science MC305-16, Pasadena CA 91125, USA, {\tt dvij@caltech.edu}.}
\and Yuval Rabani\thanks{The Rachel and Selim Benin School of Computer
Science and Engineering, The Hebrew University of Jerusalem,
Jerusalem 9190401, Israel, {\tt yrabani@cs.huji.ac.il}, supported in part by ISF grant 956-15, BSF grant 2012333 and I-CORE Algo.}
\and Leonard J. Schulman\thanks{California Institute of Technology,
Engineering and Applied Science MC305-16, Pasadena CA 91125, 
USA, {\tt schulman@caltech.edu}, supported in part by NSF grant 1319745 
and BSF grant 2012333.}
}
%\date{\today}

\maketitle

\begin{abstract}
In both general equilibrium theory and game theory, the dominant
mathematical models rest on a fully rational {\em solution concept}
in which every player's action is a best-response to the
actions of the other players. In both theories there is less agreement on suitable
{\em out-of-equilibrium} modeling, but one attractive approach is the
{\em level $k$ model} in which a level $0$ player adopts a very
simple response to current conditions, a level $1$ player best-responds
to a model in which others take level $0$ actions, and so forth.
(This is analogous to $k$-ply exploration of game trees in AI, and
to receding-horizon control in control theory.) If players have deterministic
mental models with this kind of finite-level response, there is obviously no way their mental models can all be
consistent. Nevertheless, there is experimental evidence that people
act this way in many situations, motivating the question of what the dynamics of such interactions
lead to.

We address this question in the setting of Fisher Markets with \emph{constant elasticities of substitution} (CES) utilities,
in the \emph{weak gross substitutes} (WGS) regime. We show that despite the inconsistency of the mental
models, and even if players' models change arbitrarily from round to round, the market
converges to its unique equilibrium. (We show this for both synchronous
and asynchronous discrete-time updates.) Moreover, the result is computationally
feasible in the sense that the convergence rate is linear, i.e., the distance
to equilibrium decays exponentially fast. To the best of our knowledge, this
is the first result that demonstrates, in Fisher markets, convergence at any 
rate for dynamics driven by a plausible model of seller incentives. Even
for the simple case of (level $0$) best-response dynamics, where we observe
that convergence at some rate can be derived from recent results in convex
optimization, our result is the first to demonstrate a linear rate of convergence.
\end{abstract}

\thispagestyle{empty}
\newpage
\setcounter{page}{1}

%%%%%%%%%%%%%%%%%%%%%%%%%%%%%%%%%%%%%%%%%%
%%%%%%%%%    Introduction
%%%%%%%%%%%%%%%%%%%%%%%%%%%%%%%%%%%%%%%%%%
\section{Introduction}

\paragraph{Motivation}
This paper deals with the question of why, and whether, a model
of interacting strategic agents converges to equilibrium. We
study this question in Fisher markets, under conditions where 
market equilibrium is unique and finding it is computationally tractable. Over
the years and in particular recently, several game and market
dynamics have been studied, but they fall short of modeling  
the key scenario which we attempt to address. 

In particular, in game theory,
dynamics are studied in the context of repeated games. Extensive 
form solution concepts such as subgame perfect or sequential
equilibria assume that the agents unravel the entire evolution
of the game and choose in advance their entire play optimally.
This is likely to be computationally infeasible (e.g.~\cite{BorgsCIKMP10},
but see in contrast~\cite{HalpernPS14}). The strategies
are unrealistically prescient of the distant future, contradicting
experience and hindering on-the-fly adaptation to unexpected
changes. Just as importantly, since the entire play is determined a-priori
and an equilibrium is played throughout, such concepts do not
capture out-of-equilibrium behavior (that may lead to equilibrium
over time), so they are in fact a static notion. 

Walrasian 
{\em t\^atonnement}, and more generally game theoretic 
{\em learning} dynamics (a.k.a. no-regret dynamics), are an
alternative approach. These are truly dynamic, out-of-equilibrium frameworks, that
can be shown in many cases to converge to an attractive
solution concept. However, the reactions of the agents have 
to be damped carefully for a desirable outcome to materialize;
such reactions lack strategic justification
(see~\cite{BlumM07} and the references therein, e.g.,~\cite{HartM00}).

Closer to our work, various formulations of bounded rationality
have provided a rich basis for progress in game theory and, over
the last two decades, in its algorithmic aspects. The most basic 
approach in this vein is {\em best-response} dynamics. Agents 
play myopically an optimal move at each round, assuming that the
other agents will not deviate from their 
existing strategy.
\footnote{The situations where best-response is known to lead to 
an attractive outcome are tightly connected to the concept of 
{\em potential games}. 
See~\cite{MondererS96,AwerbuchAEMS08,ChienS11}. For a damped 
version, {\em logit} dynamics, see~\cite{AulettaFPPP15}. For a general
discussion of best-response and the related {\em fictitious play}
dynamics, see~\cite{ShohamL09}.}
A strategy which is somewhat 
more sophisticated than best-response is limited-depth exploration 
of an extensive form game tree. This is an approach to complex 
games that was developed in the early days of AI (the exploration 
depth is sometimes called the {\em ply} of a search). Essentially
the same concept is known in control theory as {\em receding-horizon 
control}. This is in contrast with the full-rationality approach underlying 
solution concepts such as the aforementioned subgame perfect or 
sequential equilibria.

In game theory, the idea that people compete by pursuing limited-lookahead situational analysis goes under the rubric of the {\em level $k$ model}, initiated
 by~\cite{stahlW94,stahlW95} and~\cite{nagel95}; related ideas are also known as 
{\em cognitive hierarchy}, {\em higher-order rationality}, and 
{\em bounded depth of reasoning}. 
The idea has been subjected to many experimental 
tests---see~\cite{hoCW98,cgCB01,crawford03,camererHC04,cgC06,crawfordI07a,crawfordI07b}---and has emerged with considerable support. For recent theoretical work on the model see~\cite{strzalecki14,kneeland,dCSS,gorelkina}; for a survey see~\cite{ccgi13}.

In view of the above, it is important to study the dynamics and
stability of markets composed of agents each of whom performs
some limited lookahead and plays optimally against that forecast.
Limited lookahead means that each agent $j$ has a mental model
of each other agent $k$, where $k$ looks ahead some constant number
of steps, and based on that chooses an optimal action (according to $j$'s
perception). Based on this model, $j$ chooses a move that is optimal 
conditional on those other imagined actions. The paradox of endless 
self-reference is obvious here, and is precisely the point of the exercise: 
such a model does not make sense for infinitely-intelligent agents who 
possess perfect common knowledge of the properties of the market. 
But such agents do not exist. Instead, the model is consistent with 
experience that markets are composed of many agents who, despite 
having limited ability to predict the actions of others, do their best to 
make such a prediction and then respond optimally to their own prediction. 
This is a very different approach to agent choice than the ``solution concept" 
notion on which game theory rests: Nash equilibria, correlated equilibria, 
the core, and so forth. In particular, one difference is that in contrast with
full rationality, in the limited lookahead case the beliefs of the agents are 
not necessarily consistent with each other and with reality. In fact, they may
even be self-inconsistent across time steps. From a purely mathematical perspective these inconsistencies might appear to be a fatal flaw. We hold differently, that this is part of the challenge of modeling out-of-equilibrium strategic play. The market is out of equilibrium 
{\em because} players do not have perfect models of each other, or because they are uncertain about exogenous factors a few steps into the future. We further hold that the predictive power in experiments of the {\em level $k$ model} is ample reason to study its dynamics.
%the stability and convergence properties of markets in which it is employed. 
That is what we do here (and for a more general notion of
best-response with lookahead).

\paragraph{Our results}
This paper is devoted to studying the dynamics and stability
of markets where the agents model their
peers as using limited lookahead. We focus on one of the best-understood
cases of general equilibrium theory, namely Fisher markets that
consist of sellers of goods and buyers endowed with budgets.

In fact, we develop a general framework that shows convergence
of dynamics based on limited lookahead, only requiring certain abstract conditions on the updates of the players. A concrete special case is a Fisher market
in which the buyers generate demand due to utilities that exhibit 
{\em constant elasticity of substitution} (CES), in the {\em weak
gross substitutes} (WGS) regime. (Rigorous definitions await
Sections~\ref{sec: preliminaries} and~\ref{sec: concrete}.)
CES utilities were chosen because this setting is very well 
understood in the context of discrete-time t\^atonnement 
(see~\cite{CCD13}), and this gives us some comparative
perspective. The restriction to WGS was imposed because 
otherwise even staying at a market equilibrium cannot be 
reasoned by individual sellers best-responding to the equilibrium.
%, so intuitively motivated responses leading to equilibrium require an alternative justification.

Our dynamic model focuses on the sellers. Each seller
controls and sets the price of a unique single good. The buyers
are assumed to react instantly and myopically to current prices
by adjusting their demand to optimize their utilities subject to
their budgets. This assumption can be justified, for instance, by
assuming that there is a large number of buyers, each contributing
negligibly to the demand.
%, or by assuming that the buyers model types with a fixed flow of one-time customers, etc. 
Each seller
is assumed to form a belief on the next move of each of the
other sellers, and then to choose a price that optimizes its own
profit based on this belief. We analyze a rather general
belief formation process that includes, as a special case,
beliefs based on assuming that the other sellers use limited
lookahead. We assume neither consistency among the
beliefs formed by different sellers, nor consistency
among the beliefs formed by the same seller at different
times. This includes as a special case, but is considerably
more general than, {\em level $k$} choices. See Figure~\ref{genl-fict-play}.
We refer to dynamics of this sort as 
{\em best-response with lookahead} (abbreviated BRL) dynamics.

\begin{figure}
\includegraphics[height=150mm]{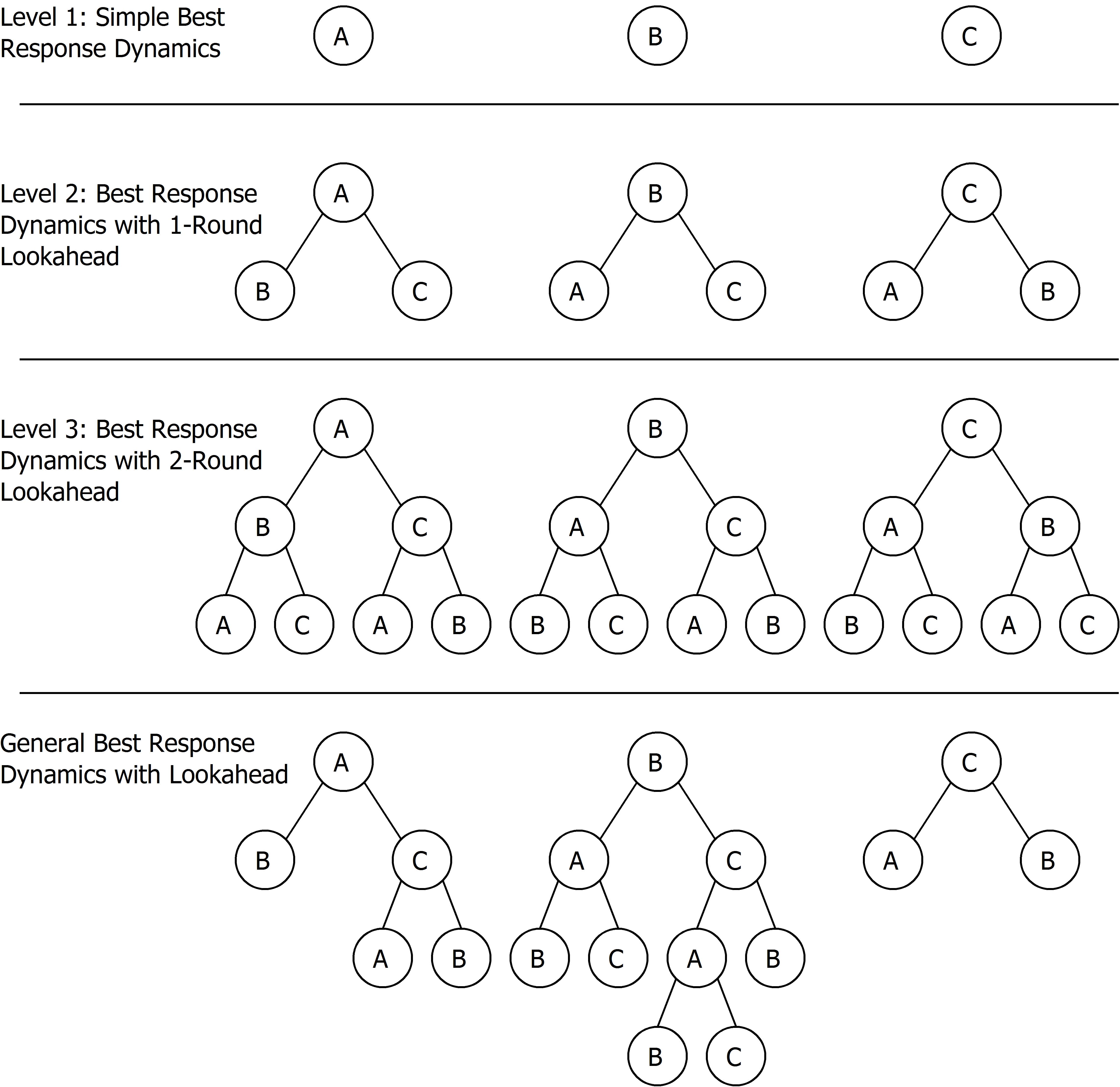}
\caption{Various collective mental models for one round of play in a 3-seller market. A leaf ({\em level $1$}) denotes best-response dynamics. In {\em level $2$} dynamics everyone 
best-responds to everyone's best-response to current prices. Players' beliefs can be far more complex. In the last example C plays by a {\em level $2$} model while A and B have more elaborate mental models.}
\label{genl-fict-play}
\end{figure}

% Bounded ply (or even just) 
BRL dynamics, and even the special case of best-response 
(with no lookahead), can be quite volatile,
as compared with usual t\^atonnement processes, because
of the absence of any damping factor. Despite the volatility and the potential
inconsistency of beliefs, we show that 
regardless of the specifics of the beliefs formed by the agents, 
the dynamic converges rapidly to market equilibrium. More precisely, we analyze two 
versions of our process. In the synchronous case, all sellers 
update prices simultaneously. In this case, the distance
to equilibrium decays exponentially in the number of steps 
(a.k.a.\ {\em linear convergence}).
In the asynchronous case, at each time step only a subset
of one or more sellers update prices. In this case, the distance 
to equilibrium decays exponentially in the number of epochs, 
where an epoch consists of time intervals in which all the sellers 
update at least once. 

To the best of our knowledge, convergence, and definitely linear
convergence, was not previously demonstrated even for the simplest 
version of our process, namely best-response. Our proof of convergence 
relies on showing that in a judiciously chosen metric (the Thompson 
metric), the BRL dynamics form a contraction map.

\paragraph{Related work}
General equilibrium theory is the principal framework through 
which economists understand the operation of markets 
(see~\cite{McKenzie02,Mukherji02}). It is one of the great
achievements of economic theory in general and of mathematical
modeling of microeconomics in particular. The theory is
largely responsible for the governing paradigm that a state
of {\em equilibrium} which the participants in economic
exchange do not wish to deviate from individually
is under mild conditions attainable~\cite{ArrowD54,McKenzie54}
(see also~\cite{Hildenbrand98}),
and that this is normally roughly the state of the economy.
This is a paradigm that can be observed ``in the field''
and also reproduced in controlled experiments, and it
lends credence and concreteness to the famed
{\em invisible hand} metaphor.

In contrast, there is less agreement
on an effective explanation as to {\em why} markets tend to
reach a state of equilibrium. This is a question about the
{\em stability} or {\em out-of-equilibrium} behavior of
markets. It is important because in reality economic
conditions are not static. They vary continually and
suffer serious ``shocks'' occasionally. So justifying an
equilibrium outcome requires a dynamic that moves
an economy at disequilibrium back to a new
equilibrium, and does so sufficiently quickly that the
periods of disequilibrium due to fluctuations are relatively
negligible (see~\cite{Dix90}). The classical mechanism
proposed to explain general equilibrium is Walrasian
t\^atonnement~\cite{Walras74}, a process that reacts
to excess demand by raising the price and to excess
supply by reducing the price. Variants of t\^atonnement
are known to converge to equilibrium, at least in some
classes of markets including those we consider here
(e.g.~\cite{Samuelson41,ABH59,CCD13}). However,
the classical view of
t\^atonnement posits the existence of an imaginary
``auctioneer'' who controls the process by announcing
prices. Recent work on the convergence of discrete-time
t\^atonnement in Fisher markets attempts to present it
as an in-market process in the context of the so-called
{\em ongoing markets}~\cite{CF08,CCR12}. However, even
this attempt requires a somewhat careful choice of the
magnitude of the price adjustment which is not motivated
by any agent considerations (aside from a common
inexplicable passion to equilibrate the economy).
Thus, the difficulty is in formulating
a theory of out-of-equilibrium behavior that makes sense
in terms of the incentives of the participants. 

In is well-known that market equilibria in Fisher markets
with CES utilities can be expressed as solutions to a
convex program, first proposed by Eisenberg and Gale
(see~\cite{JainV10}).
We observe that best-response dynamics (i.e., the simplest
example of our setting) can, in fact, be explained as a specific
implementation of coordinate descent (in the dual program).
The convergence of coordinate descent
was established in~\cite{tseng2001convergence},
without bounds on the rate. Recently, \cite{saha2013nonasymptotic} 
established a sublinear convergence rate (the distance to the optimum 
decays linearly with the number of iterations), if the objective function
satisfies some conditions. We note that the objective
function of the dual Eisenberg-Gale program satisfies these conditions.
Our general result shows a linear convergence rate (the 
distance to equilibrium decays exponentially in the number of 
iterations), and this holds in particular in the case of best-response.
To the best of our knowledge, this is not implied by previous results.

Two recent papers consider market dynamics under strategic 
behavior. Both bound the fraction of optimal welfare that is guaranteed. 
In~\cite{BabaioffLNP14}, strategic buyers play a Nash (or Bayesian) 
equilibrium in a market in which the sellers' prices are determined by  
Walrasian t\^atonnement; note that here the t\^atonnement 
is part of the {\em mechanism} defining the game, rather than the agents' strategies.
% and not a form of  repeated play. 
In~\cite{BabaioffPS15}, sellers engage in best-response 
dynamics. In this setting the market does not actually have an equilibrium,
but a fraction of the optimal welfare can be extracted by the dynamic.
In both papers the market model is quite different from ours. 

In the game 
theory setting (as opposed to markets), best-response dynamics have 
been studied extensively in recent years, mostly concerning bounds on
the quality of the play and conditions
that imply or prevent convergence to a Nash equilibrium~\cite{Mirrokni2004,Roughgarden15,FanelliFM12,EngelbergFSW13}.
The paper~\cite{NisanSVZ11} investigates conditions under which 
best-response is a fully rational strategy.

%%%%%%%%%%%%%%%%%%%%%%%%%%%%%%%%%%%%%%%%%%
%%%%%%%%%    Preliminaries
%%%%%%%%%%%%%%%%%%%%%%%%%%%%%%%%%%%%%%%%%%
\section{Preliminaries}\label{sec: preliminaries}

\paragraph{The market model}
We consider a Fisher market with $n$ perfectly divisible goods
and $m$ buyers. Each good is initially owned by a unique seller
that controls its price, and its quantity is scaled to $1$.  
The utility of that seller is the price times $\min\{\text{demand},1\}$.
The buyers
respond instantly and myopically to price changes. Thus their role
in the process is to specify in a convenient way the demands for
the goods at any given assignment of prices to those goods. This
is done as follows. Each buyer $i$ is endowed with a positive
budget $b_i$ and a utility function $u_i$ over baskets of goods
$x$.

For a price vector $p$, we write $p > 0$ to indicate that all
the prices are strictly positive. Similarly for price vectors $p,q$,
we write $p>q$ (resp.\ $p \geq q$) if $p_j>q_j$ (resp.\ $p_j\geq q_j$)
for all $j$.

Given a price vector $p > 0$, the demands for the goods are
determined as
follows. Every buyer $i$ chooses $x_i$ to optimize the utility
function $u_i(x)$, subject to the budget constaint
\begin{align}
\sum_j p_j x_{ij}\leq b_i.\label{eq:Budget}
\end{align}
We denote the utility maximizing allocations for prices $p$
by $x(p)$. The demand for each good $j$ at prices $p$
is $\sum_i x_{ij}(p)$.

\paragraph{Price updates}
In general, a market dynamic is based on an update rule for
each seller that determines its new price. The rules can then
by applied synchronously to all sellers, or serially to one
seller at a time in some order. We will discuss these variations
later. For now, we focus on the update rules. An update rule
can take into account some or all of the dynamic history leading
to the current state (including the current prices), and also
some internal state of the seller that takes other factors into
account. We are interested in update rules that depend on
the current price vector (and any other parameters),
and are {\em monotone}, {\em sub-homogeneous},
{\em price-bounded}, and {\em positive} with respect to that price vector.
To define these properties formally, let $F^{\iota}_j(p)$ denote the
new price of seller $j$, given current prices $p$, and $\iota$
encoding all the other relevant parameters (if any). Then,
\begin{definition}[monotonicity, sub-homogeneity, price-boundedness, positivity]
We say that:
\begin{itemize}
\item $F^{\iota}_j$ is monotone if for all pairs
of price vectors $p,q$ such that $p\ge q$ coordinate-wise,
$F^{\iota}_j(p)\ge F^{\iota}_j(q)$;

\item $F^{\iota}_j$ is sub-homogeneous if for all
price vectors $p$ and for all $\lambda\in (0,1)$,
$F^{\iota}_j(\lambda p)\ge \lambda F^{\iota}_j(p)$,
also $F^{\iota}_j$ is strictly sub-homogeneous iff
the inequality is strict for all $p > 0$;

\item $F^{\iota}_j$ is
$[p_{\min},p_{\max}]$-price-bounded if
for all price vectors $p\in [p_{\min},p_{\max}]^n$,
$F^{\iota}_j(p)\in [p_{\min},p_{\max}]$.

\item $F^{\iota}_j$ is positive if $p_{\min} > 0$.
\end{itemize}
\end{definition}

For a price vector $p$ and price updates $F_j$
for all $j\in [n]$, we denote by $F(p)$ the price
vector derived by applying the updates simultaneously
to $p$. We say that $F$ has a property (e.g.,
is monotone) if all of its components have this property.
\begin{lemma}\label{lm: composition}
Fix $p_{\min},p_{\max}$, and suppose that
$F:[p_{\min},p_{\max}]^n\rightarrow [p_{\min},p_{\max}]^n$
and $g:[p_{\min},p_{\max}]^n\rightarrow [p_{\min},p_{\max}]$
are both monotone, sub-homogeneous, and
$[p_{\min},p_{\max}]$-price bounded updates. Then
so is $g\circ F$. Moreover, if $g$ is strictly sub-homogeneous
and $F$ is positive, then $g\circ F$ is also strictly sub-homogeneous.
\end{lemma}

\begin{proof}
By the monotonicity of $F$, if $p \ge q$ coordinate-wise,
then $F(p) \ge F(q)$ coordinate-wise. Therefore, by
the monotonicity of $g$, we have that $g(F(p))\ge g(F(q))$.
Next,
$g(F(\lambda p)) \ge g(\lambda F(p))\ge \lambda g(F(p))$,
where the first inequality uses the
monotonicity of $g$ and the
sub-homogeneity of $F$,
 and the second inequality uses the
sub-homogeneity of $g$. Moreover, if $g$ is strictly
sub-homogeneous, then the second inequality is strict
if $F(p) > 0$, which is implied when $p > 0$ by
the assumption that $F$ is positive.
Finally, using the price boundedness of both $F$ and $g$,
if $p\in [p_{\min},p_{\max}]^n$, then
$F(p)\in [p_{\min},p_{\max}]^n$, so $g(F(p))\in [p_{\min},p_{\max}]$.
\end{proof}

\paragraph{Belief formation}
We consider dynamics where each seller updates its
price according to a belief of which prices the other sellers
will set in the next step. The beliefs that are formed by
different sellers or by the same seller at different
times need not be consistent.
We show that despite this inconsistency, the dynamics still
converge to equilibrium, assuming that the ingredients
satisfy certain properties.
In general, a belief
$\pi$ is a function that maps a pair $(p,\iota)$,
where $p$ is the current price vector and $\iota$
is the internal state of the seller,
to the believed price vector.

We now discuss a rather general framework of forming
such beliefs. This framework in particular enables the
sellers to form {\em level $k$} best-response beliefs, and
more general best-response beliefs. We haven't yet formally
defined ``best-response'', but for now it
suffices to assume that there is at our disposal a price
update called best-response.
% We will not need too many details here.
%%YR: changed
%%A concrete
The details of best-response update
are discussed in Section~\ref{sec: concrete}. Also,
in order to get some intuition on the following explanation,
it might be useful
to visualize the trees in the bottom example in
Figure~\ref{genl-fict-play}.

We explain how seller $j$ forms a belief $\pi = \pi^{\iota}$.
The idea is that seller $j$ has, for every other
seller $k$, a mental model $\iota_k$ of the update rule
that $k$ employs, and $\pi_k$ is simply the
price that $\iota_k$ generates. Of course, in order to
form $\iota_k$, seller $j$ must also imagine
seller $k$'s mental models for all $k'\ne k$ (this
includes $j$). So, we define inductively a set
of possible mental models of seller updates, and
seller $j$ simply picks each $\iota_k$ from this set.
The set ${\cal M}$ of mental models consists of levels;
${\cal M} = \bigcup_{s=0}^{\infty}{\cal M}_s$.
They are defined inductively as follows. The base
case, level $0$, is ${\cal M}_0$ that contains a single model
of staying put at the current price. 
 Inductively, a mental model or belief $\iota$ for player $j$ is formed by selecting any $\iota_{k_1},\ldots,\iota_{k_n}$ (but there is no $\iota_j$); the price update defined by this mental model is player $j$'s best-response to the prices generated by all the other players if they act with the assigned mental models on the basis of the current prices. The \textit{level} of $\iota$ is one more than the maximum level of $\iota_{k_1},\ldots,\iota_{k_n}$.

\suppress{
\ljs{Previous text is suggested replacement for the rest of this OLD TEXT:}
The inductive step
includes in ${\cal M}_s$ all the models $\iota_k$ that can
be generated as follows. For every $k'\ne k$, pick a
model $\iota_{k,k'}\in\bigcup_{r=0}^{s-1}{\cal M}_{r}$.
%%YR: I edited this sentence, mostly to fix a
%% typesetting eyesore, but also for clarity.
This defines $\iota_k$, which is $j$'s model of $k$;
according to $j$'s model, seller $k$ forms a belief
$\pi^{\iota_{k,\cdot}}$ based on the models
$\iota_{k,k'}$ for all $k'\ne k$, and updates by
best-responding to this belief.
Notice that there is an intuitive notion of ``depth'' that
can be associated with a mental model $\iota$, which
is the minimum $s$ such that for all $k$,
$\iota_k\in \bigcup_{r=0}^{s}{\cal M}_{r}$. Notice
that the depth of $j$'s model $\iota$ is greater than the
depth of any model $\iota_{k,\cdot}$ that $j$ assumes
the other sellers have.
%; this is why the beliefs must be collectively inconsistent.
\ljs{OLD TEXT TILL HERE}
}

%%YR: moved this paragraph here
We note in passing that beliefs thus formed, implicitly model
sellers with epistemic assumptions that they are a bit smarter
than their peers---every seller $j$ updates with one extra step
beyond the maximum number of steps used in $j$'s mental
model $\iota$. Of course, such beliefs cannot possibly be
consistent among sellers (unless they are children in Lake
Wobegon).

The following lemma states the desired properties of
belief formation.
\begin{lemma}\label{lm: belief-based updates}
Fix $p_{\min},p_{\max}$. Suppose that best-response
is monotone, sub-homogeneous,
$[p_{\min},p_{\max}]$-price bounded and positive. Further suppose
that seller $j$ uses a monotone, strictly sub-homogeneous,
and $[p_{\min},p_{\max}]$-price bounded update function
$F_j$, and given current prices $p$, updates to
$F^{\iota}_j(p) = F_j(\pi^{\iota}(p))$.
Then, $F^{\iota}_j$ is monotone, strictly sub-homogeneous,
%%YR: In the end, I couldn't resist the temptation to add the footnote.
and $[p_{\min},p_{\max}]$-price-bounded.\footnote{In fact, the
conclusion of Lemma~\ref{lm: belief-based updates} holds even
for beliefs that are formed by a set of monotone, sub-homogeneous,
price bounded, and positive price updates, instead of a single
such update. In the tree view of belief formation, such a set
is used by choosing, for each node of the tree, an arbitrary
member of the set as the modeled action. To simplify the exposition,
we do not elaborate on this generalization.}
\end{lemma}

\begin{proof}
Since staying put at the current price is monotone, sub-homogeneous, and
$[p_{\min},p_{\max}]$-price-bounded, a simple induction
on $s$ using Lemma~\ref{lm: composition} shows that
$\pi^{\iota}$ is monotone, sub-homogeneous, and
$[p_{\min},p_{\max}]$-price-bounded. One more application
of Lemma~\ref{lm: composition} gives the desired properties
of $F^{\iota}_j$.
\end{proof}

%Figure~\ref{fig:BeliefSpace}
%illustrates what happens when the functions are each chosen
%to be one of these two types.

%\begin{figure}
%\centering
%\includegraphics[width=.8\textwidth]{BeliefSpace.png}
%\caption{Belief space: Arrows represent players that decided to update in that round}\label{fig:BeliefSpace}
%\end{figure}

%%%%%%%%%%%%%%%%%%%%%%%%%%%%%%%%%%%%%%%%%%
%%%%%%%%%    Concrete Example
%%%%%%%%%%%%%%%%%%%%%%%%%%%%%%%%%%%%%%%%%%
\section{Concrete case of CES-WGS markets}\label{sec: concrete}

\paragraph{CES utilities}
These are utility functions of the form
\begin{align}
u_i(x) = \left(\sum_j \left(c_{ij}x_{ij}\right)^{\rho}\right)^{\frac{1}{\rho}},\label{eq:Util}
\end{align}
where $x_{ij}$ denotes the quantity of good $j$ that buyer $i$
purchased. The parameter $\rho\in (-\infty,0)\cup (0,1)$ is assumed,
for simplicity, to be uniform for all buyers. These utility functions
are known as {\em constant elasticity of substitution} (CES) utilities.
When $\rho\in (0,1)$, the goods are
{\em weak gross substitutes} (WGS). In the case of $\rho\in (-\infty,0)$,
the goods are complementary.

For CES utilities, the utility-maximizing
allocations are given explicitly by the equation
\begin{align}
x_{ij}(p) = \frac{b_i}{p_j}\cdot\frac{\left(c_{ij}/p_j\right)^{\eps}}
    {\sum_k \left(c_{ik}/p_k\right)^{\eps}},\label{eq:Optx}
\end{align}
where $\eps = \frac{\rho}{1-\rho}$.
Notice that if $\rho\in (0,1)$, then $\eps \in (0,\infty)$.
In this case, the demand satisfies the following property.\footnote{The
proof of this fundamental fact is rather trivial, but we are not aware 
of a good reference. Notice that the same proof
shows that if the CES utilities are complementary ($\eps < 0$), then 
the spending on good $j$ is monotonically increasing in $p'_j$. This 
is the motivation for considering only WGS utilities.}
\begin{lemma}\label{lm: profit monotonicity} Let the utilities be CES in the WGS regime.
Fix a price vector $p$ and a good $j$. Consider all price
vectors $p'$ with the property that for all $k\ne j$, $p'_k = p_k$.
Among these price vectors, the total desired spending
$\sum_i x_{ij}(p')\cdot p'_j$ on good $j$ is monotonically
decreasing in $p'_j$.
\end{lemma}

\begin{proof}
Using Equation~\eqref{eq:Optx},
the total desired spending on $j$ is given by the equation
$$
\sum_i x_{ij}(p')\cdot p'_j =
    \sum_i b_i\cdot\frac{\left(c_{ij}/p'_j\right)^{\eps}}
    {\left(c_{ij}/p'_j\right)^{\eps}+\sum_{k\ne j} \left(c_{ik}/p_k\right)^{\eps}}.
$$
The derivative of the right-hand side with respect to $p'_j$
is
$$
-\eps\cdot \sum_i \frac{b_i}{p'_j}\cdot\left(\left(\frac{\left(c_{ij}/p'_j\right)^{\eps}}
    {\left(c_{ij}/p'_j\right)^{\eps}+\sum_{k\ne j} \left(c_{ik}/p_k\right)^{\eps}}\right)
    - \left(\frac{\left(c_{ij}/p'_j\right)^{\eps}}
    {\left(c_{ij}/p'_j\right)^{\eps}+\sum_{k\ne j} \left(c_{ik}/p_k\right)^{\eps}}\right)^2\right).
$$
This expression is negative, because
$\frac{\left(c_{ij}/p'_j\right)^{\eps}}
    {\left(c_{ij}/p'_j\right)^{\eps}+\sum_{k\ne j} \left(c_{ik}/p_k\right)^{\eps}} < 1$.
\end{proof}

\begin{corollary}\label{cor: profit maximality}
Using the same notation as in Lemma~\ref{lm: profit monotonicity},
the profit $\min\left\{\sum_i x_{ij}(p'),1\right\}\cdot p'_j$ of seller $j$
is maximized at the price $p'_j$ for which the demand $\sum_i x_{ij}(p')$
equals $1$.
\end{corollary}

\begin{proof}
By Equation~\eqref{eq:Optx}, the demand $\sum_i x_{ij}(p')$ decreases
monotonically in $p'_j$. By Lemma~\ref{lm: profit monotonicity}, also
the desired spending decreases monotonically in $p'_j$. Therefore,
the profit is maximized at the lowest price for which the demand is
at most $1$ (lowering the price further will not increase the quantity
sold beyond the initial endowment).
\end{proof}

\paragraph{Best-response updates}
In standard best-response dynamics, each seller updates its price
to maximize its revenue given the current
prices of the other players. In the particular setting of demand
that is generated by CES utilities in the WGS regime, a seller $j$
maximizes profit by setting the price $p_j$ to clear the market for
good $j$ (by Corollary~\ref{cor: profit maximality}).
I.e., if the current price vector is $p$, the seller chooses
a new price $F_j(p)$ for good $j$, so that
\begin{align}
\sum_{i=1}^m x_{ij}(p')=1,\label{eq:Seller}
\end{align}
where $p'_j = F_j(p)$, and for all $k\ne j$, $p'_k = p_k$.
More explicitly, seller $j$ best-responds by solving
for $p'_j$ the equation
\begin{align}
p'_j = \sum_{i=1}^m b_i\cdot\frac{\left(c_{ij}/p'_j\right)^{\epsilon}}
   {\sum_{k\neq j} \left(c_{ik}/p_k\right)^{\epsilon}+\left(c_{ij}/p'_j\right)^{\epsilon}}. \label{eq:BestResponse}
\end{align}
Notice that the right-hand side of Equation~\eqref{eq:BestResponse}
is simply the total spending of all the buyers on good $j$ when this
good's price is $p'_j$ and the other prices are given by the vector
$p$.
\begin{lemma}\label{lm: best response properties}
For CES utilities with $\rho\in (0,1)$,
best-response updates $F_j$ are monotone, strictly
sub-homogeneous, positive, and $[p_{\min},p_{\max}]$-price-bounded
for some $p_{\min} = p_{\min}(b,c,\rho)>0$ and
$p_{max} = p_{\max}(b)$.
\end{lemma}

\begin{proof}
We begin with monotonicity. Consider two price vectors $p \ge q$,
and let $p'_j = F_j(p)$ and let $q'_j = F_j(q)$. Consider the function
$$
g(\alpha,p) = \alpha - \sum_{i=1}^m b_i\cdot\frac{\left(c_{ij}/\alpha\right)^{\epsilon}}
   {\sum_{k\neq j} \left(c_{ik}/p_k\right)^{\epsilon}+\left(c_{ij}/\alpha\right)^{\epsilon}}.
$$
In other words, $g(\alpha,p)$ is $\alpha$ minus the total spending of all the
buyers on good $j$ when the price of good $j$ is $\alpha$ and the other prices
are given by the vector $p$ (thus, $g(\alpha,p) = 0$ iff $\alpha = F_j(p)$).
We have that $g(p'_j,p) = g(q'_j,q) = 0$. Notice that for any $k$, increasing $p_k$
decreases $g(\alpha,p)$. On the other hand, $g(\alpha,p)$ increases as $\alpha$
increases (an immediate consequence of
Lemma~\ref{lm: profit monotonicity}). Thus, $g(p'_j,q)\geq 0$. If
$g(\alpha,q) = 0 \le g(p'_j,q)$, then it must be that $\alpha\le p'_j$.
Thus $q'_j\le p'_j$.
	
Next we prove sub-homogeneity.	Let $p'_j=F_j(p)$. Then,
$$
g(\lambda p'_j,\lambda p) =
\lambda p'_j - \sum_i  b_i \cdot \frac{\left(c_{ij}/p'_j\right)^{\epsilon}}
   {\sum_{k\neq j} \left(c_{ik}/p_k\right)^{\epsilon} +
   \left(c_{ij}/p'_j\right)^{\epsilon}} = g(p'_j,p) - (1-\lambda)\cdot p'_j \leq 0.
$$
Thus, by the monotonicity of $g$ in $\alpha$,
if $g(\alpha,\lambda p) = 0\ge g(\lambda p'_j,\lambda p)$,
then $\alpha \geq \lambda p'_j$. Finally, if all the entries of $p$
are strictly positive, then $p'_j > 0$, so
$g(\lambda p'_j,\lambda p) < 0$ and therefore
$g(\alpha,\lambda p) = 0$ implies that $\alpha > \lambda p'_j$.

%%YR: placed the statement on positivity here instead of at the end. LJS: moved to end (again?) after discussion on Nov 2.
Finally, we show price-boundedness.
% (which also trivially implies positivity, as $p_{\min} > 0$). 
As before, let $p'_j=F_j(p)$.
For the upper bound, notice that if
$p_j > \sum_i b_i$ then the total demand for good $j$ must be
less than $1$, regardless of the other prices. So we can simply
set $p_{\max} = \sum_i b_i$. For the lower bound, consider any buyer
$i$ with $c_{ij} > 0$. Consider the situation where all the entries
of $p$ are at least $p_{\min}$, and the price of good $j$ is
$q_j < p_{\min}$ (we will specify $p_{\min}$ shortly). The
demand that $i$ has for good $j$ in this situation is
$$
\frac{b_i}{q_j}\cdot\frac{\left(c_{ij}/q_j\right)^{\eps}}
{\sum_{k\neq j} \left(c_{ik}/p_k\right)^{\eps}+\left(c_{ij}/q_j\right)^{\eps}}
> \frac{C_{ij}}{p_{\min}},
$$
where
$C_{ij} = \frac{b_i\cdot c_{ij}^{\eps}}{\sum_k c_{ik}^{\eps}}$.
Set $p_{\min} = \min\{1/C_{ij}:\ c_{ij} > 0\}$. The
demand that $i$ alone generates for good $j$ is more than $1$,
so $q_j\ne p'_j$. Thus, we must have $p'_j \ge p_{\min}$. Given the value of $p_{\min}$, this shows also positivity.
\end{proof}

\paragraph{Best-response with lookahead (BRL) beliefs}
In these dynamics, each seller best-responds to
a belief (a mental model) of what the other sellers plan to do.
We already discussed a general framework of forming
beliefs. Now that we have defined best-response updates,
Lemma~\ref{lm: belief-based updates} immediately implies
the following corollary.
\begin{corollary}\label{cor: iterative best response}
Let $p_{\min},p_{\max}$ be as stipulated by
Lemma~\ref{lm: best response properties}.
Suppose that a price update $F_j$ that seller $j$ applies
to a price vector $p\in [p_{\min},p_{\max}]^n$ is a
best-response (as defined in this section) to a belief
$\pi = \pi^{\iota}(p)$ that
was generated by a mental model $\iota\in {\cal M}^{n-1}$
(using best-response as defined in this section).
Then $F_j$ is monotone, strictly sub-homogeneous, and
$[p_{\min},p_{\max}]$-price-bounded.
\end{corollary}

%%%%%%%%%%%%%%%%%%%%%%%%%%%%%%%%%%%%%%%%%%
%%%%%%%%%    Synchronous Dynamics
%%%%%%%%%%%%%%%%%%%%%%%%%%%%%%%%%%%%%%%%%%
\section{Synchronous Dynamics}
\newcommand{\br}[1]{\left({#1}\right)}
\newcommand{\norm}[1]{\left\| {#1} \right\|}\newcommand{\abr}[1]{\left | {#1} \right |}
\newcommand{\expb}[1]{\exp\left({#1} \right)}

Our main tool for proving convergence of BRL
dynamics is the following theorem. Before stating the
theorem, we require a definition.
\begin{definition}
Consider the set $\RR_{++}^n\subset \RR^n$ of vectors
with strictly positive coordinates. The
Thompson metric $d$ on $\RR_{++}^n$ (see~\cite{NonlinearPerron})
is defined as follows.
For $x,y\in\RR_{++}^n$,
\[
d(x,y) =\max_i \left| \log \frac{x_i}{y_i} \right| =
\| \log x - \log y\|_\infty,\]
% $$ d(x,y) = \log \max\left\{\max_i \frac{x_i}{y_i},\max_i \frac{y_i}{x_i}\right\} =
% \| \log x - \log y\|_\infty, $$
where $\log x$ means the vector of logarithms of the entries of $x$.
\end{definition}

The following Lemma shows how this metric is related to the standard $\ell_2$ and $\ell_\infty$ metrics.
\begin{lemma}\label{lm:MetricComparison}
Let $p^a,p^b \in [p_{\min},p_{\max}]^n$ with $0 < p_{\min} < p_{\max}$. Then, we have
\begin{align*}
& \norm{p^a-p^b}_{\infty}\leq \frac{\br{p_{\max}}^2}{p_{\min}}d\br{p^a,p^b} \\
& \norm{p^a-p^b}_{2}\leq \sqrt{n}\frac{\br{p_{\max}}^2}{p_{\min}}d\br{p^a,p^b}	
\end{align*}
\end{lemma}
\begin{proof}
For each $i$, we have
\[\abr{p^a_i-p^b_i} \leq p_{\max}\abr{\frac{p^a_i}{p^b_i}-1}\leq p_{\max} \br{\expb{d\br{p^a,p^b}}-1}\]	
The function $f\br{t}=\expb{t}-1-\kappa t$ is non-increasing on the interval $[0,\log\br{\kappa}]$ for every $\kappa>0$ and evaluates to $0$ at $t=0$. Hence $\expb{t}-1 \leq \kappa t$ for every $t \in [0,\log\br{\kappa}]$. 

Since $d\br{p^a,p^b}\leq \log\br{\frac{p_{\max}}{p_{\min}}}$, we can choose $\kappa=\frac{p_{\max}}{p_{\min}}$ and conclude that
\[\expb{d\br{p^a,p^b}}-1 \leq \frac{p_{\max}}{p_{\min}} d\br{p^a,p^b}\]
Thus, we have 
\[\abr{p^a_i-p^b_i} \leq p_{\max}\frac{p_{\max}}{p_{\min}}d\br{p^a,p^b}=\frac{\br{p_{\max}}^2}{p_{\min}}d\br{p^a,p^b}\]	
Since the bound holds for each $i$, it holds for the $\infty$ norm as well. The $2$-norm bound simply uses the fact that the $2$ norm is at most $\sqrt{n}$ times the infinity norm.
\end{proof}

\begin{theorem}\label{thm: contraction}
Let $0 < p_{\min}\le p_{\max} < \infty$.
If $F:[p_{\min},p_{\max}]^n\rightarrow [p_{\min},p_{\max}]^n$ is
monotone and strictly sub-homogeneous,
% and (evidently) $[p_{\min},p_{\max}]$-price-bounded,
then it is a contraction
with respect to the Thompson metric $d$.
%%YR: added the following.
Also, if we require merely sub-homogeneity,
rather than strict sub-homogeneity,
then $F$ is non-expanding with respect to $d$.
\end{theorem}
% \ljs{what does "evidently" mean here?}

\begin{proof}
By assumption $F$ is  $[p_{\min},p_{\max}]$-price-bounded.
Fix $p,q\in [p_{\min},p_{\max}]^n$. Let $\eta = e^{d(p,q)}$.
Then, we have $p \geq \frac{1}{\eta}\cdot q$ and $q\geq \frac{1}{\eta}\cdot p$.
We get that
$$
F(p) \geq F\left(\frac{1}{\eta}\cdot q\right) > \frac{1}{\eta}\cdot F(q),
$$
where the first inequality uses monotonicity and the second
inequality uses strict sub-homogeneity. Similarly,
$$
F(q) > \frac{1}{\eta}\cdot F(p).
$$
Thus, $d(F(p),F(q)) < \log \eta = d(p,q)$ for all $p,q\in [p_{\min},p_{\max}]^n$.
Define
$$
h(\xi) = \sup\{d(F(p),F(q)) - \xi\cdot d(p,q):\  p,q\in [p_{\min},p_{\max}]^n\}.
$$
Since $[p_{\min},p_{\max}]^n$ is a compact set, $h(1) < 0$. As $h$ is continuous,
there exists $\xi\in (0,1)$  such that $h(\xi) < 0$. Thus, $F$ is a contraction
mapping with a contraction constant $\xi < 1$.
%%YR: added the following.
Finally, if we replace strict sub-homogeneity by sub-homogeneity,
then all the strict inequalities above become weak inequalities,
so we get that $d(F(p),F(q))\le d(p,q)$, as stipulated.
\end{proof}

By Lemma~\ref{lm: best response properties}, we immediately
get the following corollary.
\begin{corollary}\label{cor: best response contraction}
An update that consists of all sellers best-responding
to the current price $p$ is a contraction. We denote
its contraction constant by $\xi_{\best} < 1$.
\end{corollary}

We say that a set ${\cal F}$ of price vector updates is {\em contracting}
if all its elements are contractions and there is a uniform upper
bound $<1$  on the contraction constants. We have the following
corollary of Theorem~\ref{thm: contraction}.
\begin{corollary}\label{cor: belief based best response contraction}
Let ${\cal F}$ be a set of price vector updates where each
$F\in {\cal F}$ is generated by all the sellers forming beliefs
that satisfy the conditions of Lemma~\ref{lm: belief-based updates},
then best-responding to those beliefs. Then, ${\cal F}$ is contracting, with
uniform bound $\le \xi_{\best}$.
\end{corollary}

\begin{proof}
Clearly, combining Lemmas~\ref{lm: belief-based updates}
and~\ref{lm: best response properties} with Theorem~\ref{thm: contraction}
proves that every $F\in {\cal F}$ is a contraction. So in order to
complete the proof, we need to show that the contraction constant
is at most $\xi_{\best}$. Let $B_j$ denote the best-response
price update of seller $j$. Notice that by the definition of $\xi_{\best}$,
for every $p,q\in [p_{\min},p_{\max}]^n$, $p\ne q$, and for every
$i,j\in [n]$,
$$
\frac{|\log(B_j(p) / B_j(q))|}{|\log(p_i/q_i)|}\le\xi_{\best}.
$$
Let $\pi^j$ denote the belief (a mapping between price vectors) that
is used by $j$ in the update $F_j$. I.e., $F_j(p) = B_j(\pi^j(p))$. Notice
that by our assumptions on $\pi^j$ and Theorem~\ref{thm: contraction},
$\pi^j$ is non-expanding. In particular, for every
$p,q\in [p_{\min},p_{\max}]^n$, $p\ne q$, and for every
$i,k\in [n]$,
$$
\frac{|\log(\pi^j_k(p) / \pi^j_k(q))|}{|\log(p_i/q_i)|}\le 1.
$$
Thus, we have
that for every $p,q\in [p_{\min},p_{\max}]^n$, $p\ne q$, and for every
$i,j,k\in [n]$,
$$
\frac{|\log(F_j(p) / F_j(q))|}{|\log(p_i/q_i)|} =
\frac{|\log(B_j(\pi^j(p)) / B_j(\pi^j(q)))|}{|\log(\pi^j_k(p)/\pi^j_k(q))|}\cdot
\frac{|\log(\pi^j_k(p) / \pi^j_k(q))|}{|\log(p_i/q_i)|}\le \xi_{\best}\cdot 1  = \xi_{\best},
$$
which completes the proof.
\end{proof}
%\ljs{ Cor.~\ref{cor: belief based best response contraction} needs proof.
%And do we still need Cor.~\ref{cor: finite set contraction} ?}
%%YR: Added a proof. Commented the trivial yet superfluous corollary.

%Also note that (trivially)
%\begin{corollary}\label{cor: finite set contraction}
%Any finite set ${\cal F}$ of price vector updates that
%are contractions is contracting.
%\end{corollary}

We need an extra property to guarantee convergence
to equilibrium.
\begin{definition}[local stability]
We say that a set of price vector updates ${\cal F}$
is {\em stable} if the following holds for every $F\in {\cal F}$: 
If $p^\ast$ is a vector of equilibrium prices, then $F(p^\ast) = p^\ast$.
\end{definition}

Our main result is the following convergence theorem.
\begin{theorem}\label{thm: synchronous main}
Fix $p_{\min},p_{\max}$.
Let ${\cal F}$ be a contracting and stable set of price
vector updates, where all $F\in {\cal F}$
are monotone, strictly sub-homogeneous, and
$[p_{\min},p_{\max}]$-price-bounded.
Consider the dynamic $p^{t+1} = F^t(p^t)$,
where the choice of $F^t \in {\cal F}$ is arbitrary.\footnote{In
particular, if ${\cal F}$ is a product set, then
for all $j$, $F^t_j$ can be chosen arbitrarily by seller
$j$. Also, the choices can depend on the entire
history of the process, including, but not limited to, the
current prices.}
Then, with initial price vector $p^0\in [p_{\min},p_{\max}]^n$,
the dynamic converges to an equilibrium point (which must
be unique in this case). Moreover, the rate of convergence
is linear (i.e., the distance to equilibrium decays
exponentially fast in the number of time steps).
\end{theorem}

\begin{proof}
Let $d$ be the Thompson metric on $[p_{\min},p_{\max}]^n$.
Consider an equilibrium point $p^*$. By the stability of
${\cal F}$, for all $t$, $F^t(p^*) = p^*$. By the fact that ${\cal F}$
is contracting, there exists $\xi_{\max} < 1$ such that for
all $t$, $F^t$ is a $\xi_{\max}$-contraction.
Therefore,
$$
d(p^{t+1},p^*) = d(F^t(p^t),F^t(p^*))\le
\xi_{\max}\cdot d(p^t,p^*).
$$
Inductively, for every $T\ge 0$,
$$
d(p^T,p^*)\le \left(\xi_{\max}\right)^T d(p^0,p^*).
$$
This completes the proof.
\end{proof}

\begin{corollary}
If the buyer utilities are CES with $\rho > 0$, then
synchronous best-response dynamics, as well as BRL
dynamics that satisfy the conditions of
Corollary~\ref{cor: iterative best response}, both converge
to equilibrium at a linear rate.
\end{corollary}

\begin{corollary}
Let $\mathcal{F}$ satisfy the assumptions of theorem \ref{thm: synchronous main}.
Let $p^\ast$ denote the vector of equilibrium prices and $p^0$ the initial price vector. After $T$ steps of synchronous price updates, let $p^T$ denote the resulting price vector. Then,
\[{\| p^T-p^\ast \|}_{\infty} = \max_i |p_i - p^\ast_i| \leq \frac{\br{p_{\max}}^2}{p_{\min}} d\br{p^0,p^{\ast}}\br{\zeta_{\max}}^T\]
\[{\| p^T-p^\ast \|}_{2} \leq \sqrt{n} \frac{\br{p_{\max}}^2}{p_{\min}} d\br{p^0,p^{\ast}}\br{\zeta_{\max}}^T\]
\end{corollary}
\begin{proof}
Apply lemma \ref{lm:MetricComparison} to theorem \ref{thm: synchronous main}. 
\end{proof}

%%%%%%%%%%%%%%%%%%%%%%%%%%%%%%%%%%%%%%%%%%
%%%%%%%%%    Asynchronous Dynamics
%%%%%%%%%%%%%%%%%%%%%%%%%%%%%%%%%%%%%%%%%%
\section{Asynchronous Dynamics}

We consider dynamics where each seller updates its price
at its own varying rate. Thus, at any given time, only a subset
of the sellers update their price. Adopting the notation from
the previous section, we can formally define the dynamic
by allowing some, but not all, of the coordinates of the
price vector updates $F^t$ to be the identity map.
The other coordinates are required to satisfy the same
conditions that are stated in Theorem~\ref{thm: synchronous main}.
We further require that no seller stays put with no update forever.
Thus we can partition the time line into epochs.
An epoch ends when all the prices are updated at least once,
and a new epoch begins in the next time step.
\begin{theorem}\label{thm: asynchronous main}
Under the assumptions stated above, the dynamic, starting
at initial state $p^0\in [p_{\min},p_{\max}]^n$,
converges to the unique equilibrium point. The rate of
convergence, measured by the number of epochs, is
linear, i.e., the distance to
equilibrium decays exponentially fast in the number
of epochs.
\end{theorem}

\begin{proof}
In an epoch, we can think of the last update of
each seller as being applied to the price vector in the
beginning of the epoch, and based on a belief that takes
into account all the previous updates in the epoch. So
replace the asynchronous process by a synchronous
process where time steps are epochs and price updates
map the prices in the beginning of an epoch to the
prices at the end of an epoch. The claim
follows by applying Lemma~\ref{lm: belief-based updates}
and Theorem~\ref{thm: synchronous main}.
\end{proof}

\begin{corollary}
If the buyer utilities are CES with $\rho > 0$, then
asynchronous best-response dynamics, as well as
asynchronous BRL dynamics, both converge
to equilibrium at a linear rate, when measured against
the number of epochs. By lemma \ref{lm:MetricComparison}, the linear convergence holds in the Thompson metric $d$, the $\ell_2$ and the $\ell_\infty$ metrics.
\end{corollary}

%\newpage

\bibliography{refs}

\begin{thebibliography}{10}

\bibitem{ABH59}
K.~J. Arrow, H.~D. Block, and L.~Hurwicz.
\newblock On the stability of the competitive equilibrium: {II}.
\newblock {\em Econometrica}, 27(1):82--109, 1959.

\bibitem{ArrowD54}
K.~J. Arrow and G.~Debreu.
\newblock Existence of equilibrium for a competitive economy.
\newblock {\em Econometrica}, 22:265--290, 1954.

\bibitem{AulettaFPPP15}
V.~Auletta, D.~Ferraioli, F.~Pasquale, P.~Penna, and G.~Persiano.
\newblock Convergence to equilibrium of logit dynamics for strategic games.
\newblock {\em Algorithmica}, pages 1--33, 2015.

\bibitem{AwerbuchAEMS08}
B.~Awerbuch, Y.~Azar, A.~Epstein, V.~S. Mirrkoni, and A.~Skopalik.
\newblock Fast convergence to nearly optimal solutions in potential games.
\newblock In {\em Proc. of the 9th ACM Conf. on Electronic Commerce}, pages
  264--273, 2008.

\bibitem{BabaioffLNP14}
M.~Babaioff, B.~Lucier, N.~Nisan, and R.~Paes~Leme.
\newblock On the efficiency of the walrasian mechanism.
\newblock In {\em Proceedings of the Fifteenth ACM Conference on Economics and
  Computation}, EC '14, pages 783--800, New York, NY, USA, 2014. ACM.

\bibitem{BabaioffPS15}
M.~Babaioff, R.~Paes~Leme, and B.~Sivan.
\newblock Price competition, fluctuations and welfare guarantees.
\newblock In {\em Proceedings of the Sixteenth ACM Conference on Economics and
  Computation}, EC '15, pages 759--776, New York, NY, USA, 2015. ACM.

\bibitem{BlumM07}
A.~Blum and Y.~Mansour.
\newblock Learning, regret minimization, and equilibria.
\newblock In N.~Nisan, T.~Roughgarden, E.~Tardos, and V.~V. Vazirani, editors,
  {\em Algorithmic Game Theory}, pages 79--102, 2007.

\bibitem{BorgsCIKMP10}
C.~Borgs, J.~T. Chayes, N.~Immorlica, A.~T. Kalai, V.~S. Mirrokni, and C.~H.
  Papadimitriou.
\newblock The myth of the folk theorem.
\newblock {\em Games and Economic Behavior}, 70(1):34--43, 2010.

\bibitem{camererHC04}
C.~F. Camerer, T.-H. Ho, and J.-K. Chong.
\newblock A cognitive hierarchy model of games.
\newblock {\em Quarterly Journal of Economics}, 119(3):861--898, 2004.

\bibitem{CCD13}
Y.~K. Cheung, R.~Cole, and N.~R. Devanur.
\newblock Tatonnement beyond gross substitutes?: gradient descent to the
  rescue.
\newblock In {\em Proc. of the 45th Ann. ACM Symp. on Theory of Computing},
  pages 191--200, 2013.

\bibitem{CCR12}
Y.~K. Cheung, R.~Cole, and A.~Rastogi.
\newblock Tatonnement in ongoing markets of complementary goods.
\newblock In {\em Proc. of the 13th Ann. ACM Conf. on Electronic Commerce},
  pages 337--354, 2012.

\bibitem{ChienS11}
S.~Chien and A.~Sinclair.
\newblock Convergence to approximate nash equilibria in congestion games.
\newblock {\em Games and Economic Behavior}, 71(2):315--327, 2011.

\bibitem{CF08}
R.~Cole and L.~Fleischer.
\newblock Fast-converging tatonnement algorithms for one-time and ongoing
  market problems.
\newblock In {\em Proc. of the 40th Ann. ACM Symp. on Theory of Computing},
  pages 315--324, 2008.

\bibitem{cgC06}
M.~A. Costa-Gomes and V.~P. Crawford.
\newblock Cognition and behavior in two-person guessing games: An experimental
  study.
\newblock {\em American Economic Review}, 96(5):1737--1768, 2006.

\bibitem{cgCB01}
M.~A. Costa-Gomes, V.~P. Crawford, and B.~Broseta.
\newblock Cognition and behavior in normal-form games: An experimental study.
\newblock {\em Econometrica}, 69(5):1193--1235, 2001.

\bibitem{crawford03}
V.~P. Crawford.
\newblock Lying for strategic advantage: Rational and boundedly rational
  misrepresentation of intentions.
\newblock {\em American Economic Review}, 93(1):133--149, 2003.

\bibitem{ccgi13}
V.~P. Crawford, M.~A. Costa-Gomes, and N.~Iriberri.
\newblock Structural models of non-equilibrium strategic thinking: Theory,
  evidence, and applications.
\newblock {\em Journal of Economic Literature}, 51(1):5--62, 2013.

\bibitem{crawfordI07b}
V.~P. Crawford and N.~Iriberri.
\newblock Fatal attraction: Salience, na\"ivet\'e, and sophistication in
  experimental ``hide-and-seek" games.
\newblock {\em American Economic Review}, 97(5):1731--1750, 2007.

\bibitem{crawfordI07a}
V.~P. Crawford and N.~Iriberri.
\newblock Level-k auctions: Can a non-equilibrium model of strategic thinking
  explain the winner's curse and overbidding in private-value auctions?
\newblock {\em Econometrica}, 75(6):1721--1770, 2007.

\bibitem{dCSS}
G.~de~Clippel, R.~Saran, and R.~Serrano.
\newblock Mechanism design with bounded depth of reasoning and small modeling
  mistakes.
\newblock Working Papers 2014-7, Brown University, Department of Economics,
  2014.
\newblock Downloaded Oct. 28, 2015.

\bibitem{Dix90}
H.~Dixon.
\newblock Equilibrium and explanation.
\newblock In J.~Creedy, editor, {\em The Foundations of Economic Thought},
  pages 356--394. Blackwell, 1990.

\bibitem{EngelbergFSW13}
R.~Engelberg, A.~Fabrikant, M.~Schapira, and D.~Wajc.
\newblock Best-response dynamics out of sync: Complexity and characterization.
\newblock In {\em Proceedings of the Fourteenth ACM Conference on Electronic
  Commerce}, EC '13, pages 379--396, New York, NY, USA, 2013. ACM.

\bibitem{FanelliFM12}
A.~Fanelli, M.~Flammini, and L.~Moscardelli.
\newblock The speed of convergence in congestion games under best-response
  dynamics.
\newblock {\em ACM Trans. Algorithms}, 8(3):25:1--25:15, July 2012.

\bibitem{gorelkina}
O.~Gorelkina.
\newblock The expected externality mechanism in a level-k environment.
\newblock MPI Collective Goods Preprint, No. 2015/3. Available at
  http://dx.doi.org/10.2139/ssrn.2550085; downloaded Oct. 28, 2015., 2015.

\bibitem{HalpernPS14}
J.~Y. Halpern, R.~Pass, and L.~Seeman.
\newblock Not just an empty threat: Subgame-perfect equilibrium in repeated
  games played by computationally bounded players.
\newblock In {\em Proc. of the 10th Int'l Conf. on Web and Internet Economics},
  pages 249--262, 2014.

\bibitem{HartM00}
S.~Hart and A.~Mas-Colell.
\newblock A simple adaptive procedure leading to correlated equilibrium.
\newblock {\em Econometrica}, 68:1127--1150, 2000.

\bibitem{Hildenbrand98}
W.~Hildenbrand.
\newblock An exposition of {W}ald's existence proof.
\newblock In E.~Dierker and K.~Sigmund, editors, {\em Karl Menger}, pages
  51--61. Springer, 1998.

\bibitem{hoCW98}
T.-H. Ho, C.~Camerer, and K.~Weigelt.
\newblock Iterated dominance and iterated best response in experimental
  ‘p-beauty contests’.
\newblock {\em American Economic Review}, 88(4):947--969, 1998.

\bibitem{JainV10}
K.~Jain and V.~V. Vazirani.
\newblock Eisenberg-gale markets: Algorithms and game-theoretic properties.
\newblock {\em Games and Economic Behavior}, 70(1):84--106, 2010.

\bibitem{kneeland}
T.~Kneeland.
\newblock Identifying higher-order rationality.
\newblock {\em Econometrica}, 83(5):2065--2079, 2015.

\bibitem{NonlinearPerron}
B.~Lemmens and R.~Nussbaum.
\newblock {\em Nonlinear Perron-Frobenius Theory}.
\newblock Cambridge University Press, 2012.
\newblock Cambridge Books Online.

\bibitem{McKenzie54}
L.~McKenzie.
\newblock {On equilibrium in Graham's model of world trade and other
  competitive systems}.
\newblock {\em Econometrica}, 22:147--161, 1954.

\bibitem{McKenzie02}
L.~W. McKenzie.
\newblock {\em Classical general equilibrium theory}.
\newblock MIT Press, 2002.

\bibitem{Mirrokni2004}
V.~S. Mirrokni and A.~Vetta.
\newblock Convergence issues in competitive games.
\newblock In K.~Jansen, S.~Khanna, J.~D.~P. Rolim, and D.~Ron, editors, {\em
  Proceedings APPROX and RANDOM}, LNCS 3122, pages 183--194. Springer, 2004.

\bibitem{MondererS96}
D.~Monderer and L.~S. Shapley.
\newblock Potential games.
\newblock {\em Games and Economic Behavior}, 14(1):124--143, 1996.

\bibitem{Mukherji02}
A.~Mukherji.
\newblock {\em An introduction to general equilibrium analysis}.
\newblock Oxford U Press, 2002.

\bibitem{nagel95}
R.~Nagel.
\newblock Unraveling in guessing games: An experimental study.
\newblock {\em American Economic Review}, 85(5):1313--1326, 1995.

\bibitem{NisanSVZ11}
N.~Nisan, M.~Schapira, G.~Valiant, and A.~Zohar.
\newblock Best-response mechanisms.
\newblock In {\em Innovations in Computer Science - {ICS} 2010, Tsinghua
  University, Beijing, China, January 7-9, 2011. Proceedings}, pages 155--165,
  2011.

\bibitem{Roughgarden15}
T.~Roughgarden.
\newblock Intrinsic robustness of the price of anarchy.
\newblock {\em J. ACM}, 62(5):32:1--32:42, November 2015.

\bibitem{saha2013nonasymptotic}
A.~Saha and A.~Tewari.
\newblock On the nonasymptotic convergence of cyclic coordinate descent
  methods.
\newblock {\em SIAM Journal on Optimization}, 23(1):576--601, 2013.

\bibitem{Samuelson41}
P.~A. Samuelson.
\newblock The stability of equilibrium: Comparative statics and dynamics.
\newblock {\em Econometrica}, 9:97--120, 1941.

\bibitem{ShohamL09}
Y.~Shoham and K.~Leyton-Brown.
\newblock {\em Multiagent Systems: Algorithmic, Game-Theoretic, and Logical
  Foundations}.
\newblock Cambridge University Press, 2009.

\bibitem{stahlW94}
D.~O. Stahl and P.~W. Wilson.
\newblock Experimental evidence on players' models of other players.
\newblock {\em Journal of Economic Behavior and Organization}, 25(3):309--327,
  1994.

\bibitem{stahlW95}
D.~O. Stahl and P.~W. Wilson.
\newblock On players' models of other players: Theory and experimental
  evidence.
\newblock {\em Games and Economic Behavior}, 10(1):218--254, 1995.

\bibitem{strzalecki14}
T.~Strzalecki.
\newblock Depth of reasoning and higher order beliefs.
\newblock {\em Journal of Economic Behavior \& Organization}, 108:108--122,
  2014.

\bibitem{tseng2001convergence}
P.~Tseng.
\newblock Convergence of a block coordinate descent method for
  nondifferentiable minimization.
\newblock {\em Journal of optimization theory and applications},
  109(3):475--494, 2001.

\bibitem{Walras74}
L.~Walras.
\newblock {\em El\'ements d'Economie Politique Pure}.
\newblock Corbaz, 1874.
\newblock (1st ed.\ 1874; revised ed.\ 1926; Transl.\ W. Jaff\'{e}, Elements of
  Pure Economics, Irwin, 1954).

\end{thebibliography}
\bibliographystyle{plain}
\end{document}